\documentclass[10pt,reqno]{amsart}

\usepackage{amsmath,amssymb,amsfonts,latexsym,amsthm}

\usepackage{mathrsfs}

\numberwithin{equation}{section}

\usepackage{graphicx}

\usepackage{ifthen} 

\provideboolean{shownotes} 
\setboolean{shownotes}{false} 
\newcommand{\margnote}[1]{
\ifthenelse{\boolean{shownotes}}%
{\marginpar{\raggedright\tiny\texttt{#1}}}%
{}%
}

\newcommand{\hole}[1]{
\ifthenelse{\boolean{shownotes}}%
{\begin{center} \fbox{ \rule {.25cm}{0cm}
\rule[-.1cm]{0cm}{.4cm} \parbox{.85\textwidth}{\begin{center}
\texttt{#1}\end{center}} \rule {.25cm}{0cm}}\end{center}}
{}
}

\theoremstyle{plain}

\newtheorem{theo}{Theorem}[section]
\newtheorem{lema}{Lemma}

\newtheorem{coro}{Corollary}
\newtheorem{rem}{Remark}
\theoremstyle{remark}

\theoremstyle{remark}

\begin{document}

\title[Non-hyperbolicity of the Cattaneo-Christov system]{Non-hyperbolicity of the inviscid Cattaneo-Christov system for compressible fluid flow in several space dimensions}

\author[F. Angeles]{Felipe Angeles}
 
\address{{\rm (F. Angeles)} Posgrado en Ciencias Matem\'{a}ticas\\Universidad Nacional Aut\'{o}noma 
de M\'{e}xico\\Circuito Exterior s/n, Ciudad de M\'{e}xico C.P. 04510 (Mexico)}

\email{teojkd@ciencias.unam.mx}

\begin{abstract}
We consider the coupling between the equations of motion of a compressible fluid in two and three space dimensions with Christov's equation for the heat flux. Christov's equation is a frame indifferent formulation of the classical model of Cattaneo that allows for the heat flux to be eliminated to obtain a single hyperbolic equation for the temperature. The obtained system is written in quasilinear form for the state variables density, velocity, temperature and  heat flux. It is then shown that this system is not of hyperbolic type as consequence of the presence of the Lie-Oldroyd upper-convected material derivative involved in Christov's formulation.
\end{abstract}
\maketitle

\setcounter{tocdepth}{1}

\section{Introduction}
Consider the local form of the conservation of mass, the balance of linear momentum and the balance of total energy ($E=\frac{1}{2}\rho|u|^{2}+e$) for a compressible fluid in space, namely, 
\begin{eqnarray}
\rho_{t}+\nabla\cdot(\rho v)&=&0,\label{eq:1}\\
(\rho v)_{t}+\nabla\cdot(\rho v\otimes v)&=&\nabla\cdot\mathbb{T},\label{eq:2}\\
(\rho E)_{t}+\nabla\cdot(\rho E v)&=&\nabla\cdot(\mathbb{T}v)-\nabla\cdot q,\label{eq:3}
\end{eqnarray}
together with the following set of thermodynamical and constitutive assumptions:
\begin{itemize}
	\item [(\textbf{C1})] The Cauchy stress tensor $\mathbb{T}(x,t)$ is of the form 
	\[\mathbb{T}=2\mu\mathbb{D}(v)+\lambda\nabla\cdot v\mathbb{I}-p\mathbb{I},\]
	where $v=v(x,t)$ is the velocity field, $p(x,t)$ is the thermodynamic pressure field, $\mathbb{D}=\frac{1}{2}\left(\nabla u+\nabla u^{\top}\right)$ is the deformation tensor and $\mathbb{I}$ denotes the identity matrix in $\mathbb{R}^{d}$. The quantities $\mu(x,t)$ and $\lambda(x,t)$ are the shear and bulk viscosities respectively.
	\item [(\textbf{C2})] The independent thermodynamical variables are the mass density field $\rho(x,t)>0$ and the absolute temperature field $\theta(x,t)>0$. They vary within the domain
	\[\mathcal{D}:=\left\lbrace(\rho,\theta)\in\mathbb{R}^{2}:\rho>0,\theta>0\right\rbrace.\]
	The thermodynamic pressure, $p$, the internal energy, $e$, and the coefficients $\mu$ and $\lambda$ are given smooth functions of $(\rho,\theta)$ whenever $\rho>0$ and $\theta>0$, 
	\[p,e,\lambda,\mu,\in\mathcal{C}^{\infty}(\mathcal{D}).\]
	\item [(\textbf{C3})] The viscosity coefficients and the heat conductivity satisfy the inequalities 
	\[\mu,~\frac{2}{3}\mu+\lambda\geq 0.\]
	for all $(\rho,\theta)\in\mathcal{D}$.
	\item [(\textbf{C4})] The fluid satisfies the following conditions 
	\[p>0,~p_{\rho}>0,~p_{\theta}>0,~e_{\theta}>0,~\mbox{for all}~(\rho,\theta)\in\mathcal{D}.\]
	\item [(\textbf{C5})] The heat flux vector field is related to the temperature field by Fourier's law
	\[q=-\kappa\nabla\theta\]
	where $\kappa>0$ is the heat conductivity coefficient and is assumed to be a given function of $(\rho,\theta)$ such that 
	\begin{align*}
	\kappa&\in\mathcal{C}^{\infty}(\mathcal{D}),\\
	\kappa\geq 0&\quad\forall(\rho,\theta)\in\mathcal{D}.
	\end{align*}
\end{itemize}
 Assumption \textbf{C4}, refers to compressible fluids satisfying the standard assumptions of Weyl \cite{weyl}. $p_{\rho}>0$ states that, adiabatic increase of pressure implies compression, $p_{\theta}>0$ is a Generalized Gay-Lussac's law and $e_{\theta}>0$ states that an increase of temperature at a constant volume yields an increase of internal energy.\\
The model obtained from equations \eqref{eq:1}-\eqref{eq:3} with the assumptions \textbf{C1}-\textbf{C5} is a closed system of equations for the variables $(\rho,v,\theta)$ known as the Navier-Fourier-Stokes system for compressible fluid flow. Assumption \textbf{C1} together with Fourier's law are enough to assure that this constitutive model complies with the axiom of material frame-indifference (see, \cite{gonzalez}). In particular, by setting $\mu=\lambda=\kappa=0$ in this model, we get Euler's equations for a compressible fluid. As a consequence of assumption \textbf{C4}, the mathematical structure of Euler's equations is that of a \emph{hyperbolic system of conservation laws} (\cite{daf}, \cite{serre}, \cite{smo}).\\
One of the major drawbacks of using Fourier's law to close the system \eqref{eq:1}-\eqref{eq:3} is that, it brings up an intrinsic hypothesis about the heat transfer, namely, the instantaneous response of the heat flux to a temperature gradient (\cite{kakac}). This unrealistic assumption predicts infinite speed of propagation of heat. In fact, Fourier's law together with the first law of thermodynamics ( eq. \eqref{eq:3}) can be used to obtain a parabolic equation for the temperature, \emph{the heat equation}.\\
Fourier's law may fail in heat conduction problems at the nanoscale or in problems involving very short timescales. See chapter 7 in \cite{zhang} for a general description of these situations; \cite{liutan} for a situation in high energy laser technology and \cite{vad} for another in nano-fluid heat transport. Furthermore, we can find materials exhibiting a relatively long thermal relaxation time $\tau$ (\cite{prez}, \cite{kakac}, \cite{zhang}). This means that the response of the heat flux to the temperature gradient throughout such material, is no longer small enough to be neglected and to be deemed as instantaneous. Such is the case for biological tissue (\cite{dai}, \cite{mitra}), for instance.\\
Therefore, a non-Fourier heat conduction model that predicts finite speed of propagation of heat should be addressed in situations in which Fourier's law is no longer viable.\\
The so called \emph{hyperbolic heat conduction theories} propose an evolution equation for the heat flux $q$, instead of Fourier's law, to close the system \eqref{eq:1}-\eqref{eq:3}. In such theories, heat transfer is modeled as a thermal wave. These theories have the special feature that the equation for the heat flux $q$ is regarded as both, a constitutive equation and a dynamical variable. One of the best known substitutes for Fourier's law is the Maxwell-Cattaneo law (\cite{cat})
\[\tau q_{t}+q=-\kappa\nabla\theta.\]
The coupling between Maxwell-Cattaneo's law and different models in continuum mechanics has been fairly studied. For example, Cattaneo's law has been coupled with the nonisentropic Euler-Maxwell system in plasma physics (see \cite{kawa2}). Such coupling yields a hyperbolic system of balance laws with a nonsymmmetric relaxation term, called the Euler-Maxwell-Cattaneo system. In another example, Cattaneo's law is used instead of Fourier's law in thermoelasticity, resulting in a new theory called \emph{extended thermoelasticity} (\cite{chandrasek}). In particular, global existence of smooth solution to the equations of one-dimensional nonlinear thermoelasticity when Fourier's law is replaced by Cattaneo's law has been established (see \cite{tara}, \cite{qin}). In addition, when thermal memory is considered, exponential stability of solutions for these equations has been shown before (\cite{qin}).\\
However, Maxwell-Cattaneo's law has proven to be an unsuccessful model for thermal waves in compressible fluid dynamics. To begin with, it has been shown to be incompatible with the second law of thermodynamics (\cite{rubin}). Furthermore, although Maxwell-Cattaneo's law accounts for finite speed of heat conduction, rises another problem when is used to describe heat conduction in a moving solid. As Christov and Jordan have shown, Maxwell-Cattaneo's law is not Galilean invariant (\cite{chjo}). In consecuence, this law is not compatible with the frame indifference principle in continuum dynamics (\cite{gonzalez}). For this reason, Cattaneo type laws involving \emph{objective derivatives}, i.e. derivatives independent of the frame of reference, also have been used to describe thermal waves in fluids (see \cite{lebon}, \cite{franchi} and \cite{straughan2} for example).\\
Christov and Jordan have argued the importance of having convective terms involved in the derivative of the Maxwell-Cattaneo law \cite{chjo}. They showed that, by replacing the partial time derivative with the standard material derivative, namely
\begin{equation}
\tau\left(q_{t}+u\cdot\nabla q\right)+q=-\kappa\nabla\theta,
\label{eq:4}
\end{equation}
a frame indifferent formulation of the Maxwell-Cattaneo law is obtained, one that predicts finite speed of propagation of heat. Yet, Christov and Jordan regarded such formulation as undesirable. The reason for this (they argue), is that such model was shown to be irreducible in several space dimensions. Meaning that, the heat flux cannot be eliminated in order to obtain a single equation for the temperature field (\cite{chjo}). This led Christov to reformulate the Maxwell-Cattaneo law one more time \cite{christov}. In his new model, the partial time derivative of the Maxwell-Cattaneo law is replaced by the so called, \emph{Lie-Oldroyd upper-convected derivative}, yielding the following evolution equation for the heat flux
\begin{equation}
\tau\left[q_{t}+v\cdot\nabla q-q\cdot\nabla v+(\nabla\cdot v)q\right]+q=-\kappa\nabla\theta.\label{eq:5}
\end{equation}
Observe that, in one space dimension the equations \eqref{eq:4} and \eqref{eq:5} are the same. Christov showed that, in the absence of all thermal sources or sinks and neglecting internal dissipation, this proposed upper-convected form of the Maxwell-Cattaneo law, together with the balance of internal energy, allow one to eliminate $q$, thus yielding a single equation for the temperature field.
Christov points out \cite{christov} that previous attempts of coupling a Cattaneo type model with the equations describing the dynamics of a fluid have led to questionable physical results.\\
The coupling between the basic equations in fluid dynamics, eqs. \eqref{eq:1}-\eqref{eq:3}, with the Cattaneo-Christov model, \eqref{eq:5}, has been studied in many different settings. For example, to study thermal convection in a horizontal layer of incompressible Newtonian fluid \cite{stra2}; in the propagation of acoustic and thermal waves in compressible inviscid Newtonian fluids \cite{jor}, \cite{stra} and in thermohaline convection \cite{stra4}. We refer to such systems of equations as \emph{Cattaneo-Christov systems}.\\
As Christov points out \cite{christov}, this model has to be tested. This is the precise objective of the present contribution. We will show that, the coupling between the equations for a compressible inviscid Newtonian fluid and the Cattaneo-Christov heat flux equation is unfit to model the propagation of thermal and acoustic waves in several space dimensions. It will be shown that the respective quasilinear system of equations is not hyperbolic.
\section{Quasilinear Cattaneo-Christov system for a compressible inviscid fluid}
We denote $U=(\rho,v,\theta,q)^{T}\in\mathcal{O}\subset\mathbb{R}^{N}$ as the vector of state variables, defined on the convex open set 
\[\mathcal{O}:=\left\lbrace(\rho,v,\theta,q)\in\mathbb{R}^{N}:\rho>0,\theta>0\right\rbrace,\]
where $N=2d+2$ and $d$ is the spatial dimension. We refer to $\mathcal{O}$ as the \emph{state space}. We write the equations \eqref{eq:1}-\eqref{eq:3} together with the equation \eqref{eq:5} as a quasilinear system of partial differential equations for the state variable $U$ in three space dimensions, that is, we set $d=3$. Such system has the form 
\begin{equation}
A^{0}(U)U_{t}+A^{i}(U)\partial_{i}U+Q(U)=0,
\label{eq:21}
\end{equation}
where repeated index notation has been used in the space derivatives $\partial_{i}$ and $i=1,2,3$. Here, once $U\in\mathcal{O}$ is given, each coefficient $A^{0}(U)$, $A^{i}(U)$ is a matrix of order $N\times N$ and $Q(U)$ is a vector in $\mathbb{R}^{N}$. Observe that, since $d=3$, $N=8$. To obtain the quasilinear form we provide the symbol
\begin{equation}
A(\xi;U):=\sum_{i=1}^{3}A^{i}(U)\xi_{i},\label{eq:22}
\end{equation}
for each $\xi=(\xi_{1},\xi_{2},\xi_{3})\in\mathbb{S}^{2}$, as well as $A^{0}(U)$ and $Q(U)$. Therefore, $A^{0}(U)$ is a diagonal matrix given as
\[
A^{0}(U)=\left(\begin{array}{cccc}
	1&&&\\
	&\rho\mathbb{I}_{3}&&\\
	&&\rho e_{\theta}&\\
	&&&\tau\mathbb{I}_{3}
\end{array}\right),\]
where $\mathbb{I}_{3}$ denotes the identity matrix of order $3\times 3$ and all the empty spaces refer to zero block matrices of the appropriate sizes.
\small
\begin{equation*}
A(\xi;U)=\left(\begin{array}{cccccccc}
	\xi\cdot v&\xi_{1}\rho&\xi_{2}\rho&\xi_{3}\rho&0&0&0&0\\
	\xi_{1}p_{\rho}&\rho\xi\cdot v&0&0&\xi_{1}p_{\theta}&0&0&0\\
	\xi_{2}p_{\rho}&0&\rho\xi\cdot v&0&\xi_{2}p_{\theta}&0&0&0\\
	\xi_{3}p_{\rho}&0&0&\rho\xi\cdot v&\xi_{3}p_{\theta}&0&0&0\\
	0&\xi_{1}\theta p_{\theta}&\xi_{2}\theta p_{\theta}&\xi_{3}\theta p_{\theta}&\rho e_{\theta}\xi\cdot v&\xi_{1}&\xi_{2}&\xi_{3}\\
	0&&&&\xi_{1}\kappa&\tau\xi\cdot v&0&0\\
	0&&\tau\mathcal{Q}(\xi;q)&&\xi_{2}\kappa&0&\tau\xi\cdot v&0\\
	0&&&&\xi_{3}\kappa&0&0&\tau\xi\cdot v\\
\end{array}\right),
\end{equation*}
\normalsize
where, for each $\xi\in\mathbb{S}^{2}$, the matrix $A(\xi,U)$ has a sub-block matrix, $\mathcal{Q}(q;\xi)$, of order $3\times 3$, that is being multiplied by $\tau\neq 0$, given as
\begin{equation}
\mathcal{Q}(q;\xi)=\left(\begin{array}{ccc}
	-\xi_{2} q_{2}-\xi_{3} q_{3}&\xi_{2}q_{1}&\xi_{3}q_{1}\\
	\xi_{1}q_{2}&-\xi_{1}q_{1}-\xi_{3}q_{3}&\xi_{3}q_{2}\\
	\xi_{1}q_{3}&\xi_{2}q_{3}&-\xi_{1}q_{1}-\xi_{2}q_{2}
\end{array}\right).
\label{eq:23}
\end{equation}
Finally
\[Q(U)=\left(0,0,0,0,0,q_{1},q_{2},q_{3}\right)^{\top}.\]
The resulting quasilinear system defined through these matrices is called the \emph{inviscid Cattaneo-Christov system for compressible fluid flow}. Throughout this work we simply refer to this system as the inviscid Cattaneo-Christov system.\\
Notice that, the sub-block $\tau\mathcal{Q}(q;\xi)$ contains all the terms of the Lie-Oldroyd upper convected derivative that differ from the standard material derivative. This matrix plays a key role in the ill-behavior of the inviscid Cattaneo-Christov system as it will be shown in the next section.
\section{Non-hyperbolicity}
We remind the reader that, the system \eqref{eq:21} is called \emph{hyperbolic} if for any fixed $U\in\mathcal{O}$ and $\xi\in\mathbb{S}^{2}$ the matrix $A^{0}(U)$ is non singular and the eigenvalue problem 
\[\left(A(\xi;U)-\eta A(U)\right)V=0\]
has real eigenvalues ($\eta\in\mathbb{R}$) and $N$($=8$) linearly independent eigenvectors. Hyperbolicity of a first order quasilinear system is a fundamental property which guarantees the well-posedness of the physical problem (see \cite{otto}). Moreover, for the non-characteristic unconstrained Cauchy problem, hyperbolicity is also a necessary condition for well-posedness (cf. \cite{guy}).\\
In this section we will show that, the inviscid Cattaneo-Christov system is not hyperbolic. To this end, the following lemma will be required.
\begin{lema}
There exists vectors $q=(q_{1},q_{2},q_{3})^{\top}\in\mathbb{R}^{3}$ and wave numbers $\xi\in\mathbb{S}^{2}$ such that, once fixed 
\[\operatorname{rank}{\mathcal{Q}(\xi;q)}=2\quad\mbox{and}\quad \dim\ker\mathcal{Q}(\xi;q)=1.\]
\end{lema}
\begin{proof}
Let us take $\bar{q}=(\bar{q}_{1},\bar{q}_{2},\bar{q}_{3})\in\mathbb{R}^{3}$ any vector that satisfies $\bar{q}_{i}\neq \bar{q}_{j}$ and $\bar{q}_{i}\neq 0$ for all $i,j=1,2,3$ and $\bar{\xi}\in\mathbb{S}^{2}$ as $\bar{\xi}=\frac{\bar{q}}{|\bar{q}|}$. We will show the required properties for $\bar{q}$ and $\bar{\xi}$.\\
First, observe that for every $\xi\in\mathbb{S}^{2}$, the columns of $\mathcal{Q}(\xi;\bar{q})$ are linearly dependent, since 
\begin{equation}
\bar{q}_{1}\left(\begin{array}{c}
	-\xi_{2}\bar{q}_{2}-\xi_{3}\bar{q}_{3}\\
	\xi_{1}\bar{q}_{2}\\
	\xi_{1}\bar{q}_{3}
\end{array}\right)+\bar{q}_{2}\left(\begin{array}{c}
	\xi_{2}\bar{q}_{1}\\
	-\xi_{1}\bar{q}_{1}-\xi_{3}\bar{q}_{3}\\
	\xi_{2}\bar{q}_{3}
\end{array}\right)+\bar{q}_{3}\left(\begin{array}{c}
	\xi_{3}\bar{q}_{1}\\
	\xi_{3}\bar{q}_{2}\\
	-\xi_{1}\bar{q}_{1}-\xi_{2}\bar{q}_{2}\end{array}\right)=0,
	\label{eq:31}
	\end{equation}
	and $\bar{q}\neq 0$. Then, $\operatorname{rank}{\mathcal{Q}(\xi;\bar{q})}<3$ for all $\xi\in\mathbb{S}^{2}$. Now, consider the columns of $\mathcal{Q}(\bar{\xi},\bar{q})$, 
	\[\frac{1}{|\bar{q}|}\left(\begin{array}{c}
	-\bar{q}_{2}^{2}-\bar{q}_{3}^{2}\\
	\bar{q}_{1}\bar{q}_{2}\\
	\bar{q}_{1}\bar{q}_{3}
\end{array}\right),~ \frac{1}{|\bar{q}|}\left(\begin{array}{c}
	\bar{q}_{2}\bar{q}_{1}\\
	-\bar{q}_{1}^{2}-\bar{q}_{3}^{2}\\
	\bar{q}_{2}\bar{q}_{3}
\end{array}\right),~\mbox{and}~ \frac{1}{|\bar{q}|}\left(\begin{array}{c}
	\bar{q}_{3}\bar{q}_{1}\\
	\bar{q}_{3}\bar{q}_{2}\\
	-\bar{q}_{1}^{2}-\bar{q}_{2}^{2}\end{array}\right).\]
	Each of these columns is different from zero and every pair of them is a linearly independent set. Take for example, the first two columns and assume that there is an $\alpha\in\mathbb{R}$ such that 
\[\alpha\left(\begin{array}{c}
	-\bar{q}_{2}^{2}-\bar{q}_{3}^{2}\\
	\bar{q}_{1}\bar{q}_{2}\\
	\bar{q}_{1}\bar{q}_{3}
\end{array}\right)=\left(\begin{array}{c}
	\bar{q}_{2}\bar{q}_{1}\\
	-\bar{q}_{1}^{2}-\bar{q}_{3}^{2}\\
	\bar{q}_{2}\bar{q}_{3}
\end{array}\right).\]
From the third component of this equality we get that, 
\[\alpha=\frac{\bar{q}_{2}}{\bar{q}_{1}}\]
and by using this value in the first and second equalities we obtain that $|\bar{q}|=0$, a contradiction. So, the first and second columns of $\mathcal{Q}(\bar{\xi};\bar{q})$ are linearly independent. By using the exact same argument we can show that any other pair of columns has the same property. In consequence, 
\[\operatorname{rank}{\mathcal{Q}(\bar{\xi};\bar{q})}=2.\]
This means that the dimension of the range of $\mathcal{Q}(\bar{\xi};\bar{q})$ as a linear transformation  is two. Thus, the dimension theorem implies that 
\[\dim\ker\mathcal{Q}(\bar{\xi};\bar{q})=1.\]
\end{proof}
\begin{rem}
Notice that, \eqref{eq:31} is the same as $\mathcal{Q}(\bar{\xi};\bar{q})\bar{q}=0$. Hence, the previous lemma assures that once $\bar{q}$ is fixed, the singleton $\left\lbrace \bar{q}\right\rbrace$ is a basis for $\ker\mathcal{Q}(\bar{\xi};\bar{q})$ with $\bar{\xi}=\frac{\bar{q}}{|\bar{q}|}$. Thus, the following identity is satisfied
\begin{equation}
\ker\mathcal{Q}(\bar{\xi};\bar{q})=\left\lbrace\beta(\bar{q}_{1},\bar{q}_{2},\bar{q}_{3}):\beta\in\mathbb{R}\right\rbrace.
\label{eq:32}
\end{equation}
\end{rem}

\begin{theo}
\label{nonhyp}
Under the thermodynamical assumptions \textbf{C2}-\textbf{C4}, the inviscid Cattaneo-Christov system is not hyperbolic. 
\end{theo}
\begin{proof}
Let $U\in\mathcal{O}$ and $\xi\in\mathbb{S}^{2}$. Consider the eigenvalue problem
\[\left(A(\xi;U)-\eta A^{0}(U)\right)V=0,\]
for $\eta\in\mathbb{R}$ and $V\in\mathbb{R}^{8}$. In order to prove that the inviscid Cattaneo-Christov system is not hyperbolic we will show that, there are values of $U\in\mathcal{O}$ and $\xi\in\mathbb{S}^{2}$ such that this eigenvalue problem does not have a complete set of eigenvectors.\\
First, we compute the eigenvalues for this system, that is, we look the roots of the equation 
\begin{equation}
\operatorname{det}\left(A(\xi;U)-\eta A^{0}(U)\right)=0.
\label{eq:33}
\end{equation}
For this, we use the formula to compute the determinant of a block matrix (see \cite{fuz}), that is, 
\[\operatorname{det}\left(\begin{array}{cc}
	\mathcal{A}&\mathcal{B}\\
	\mathcal{C}&\mathcal{D}
\end{array}\right)=(\operatorname{det} \mathcal{A})\operatorname{det}( \mathcal{D}-\mathcal{C}\mathcal{A}^{-1}\mathcal{B}),\]
whenever $\mathcal{A}$ is invertible. All the block matrices are given as
\begin{equation*}
\mathcal{A}=\left(\begin{array}{cccc}
	v\cdot\xi-\eta&\xi_{1}\rho&\xi_{2}\rho&\xi_{3}\rho\\
	\xi_{1}p_{\rho}&\rho v\cdot\xi-\eta\rho&0&0\\
	\xi_{2}p_{\rho}&0&\rho v\cdot\xi-\eta\rho&0\\
	\xi_{3}p_{\rho}&0&0&\rho v\cdot\xi-\eta\rho
\end{array}\right),
\end{equation*}
\begin{equation*}
\mathcal{B}=\left(\begin{array}{cccc}
	0&0&0&0\\
	\xi_{1}p_{\theta}&0&0&0\\
	\xi_{2}p_{\theta}&0&0&0\\
	\xi_{3}p_{\theta}&0&0&0
	\end{array}\right),
\end{equation*}
\begin{equation*}
\mathcal{C}=\left(\begin{array}{cccc}
	0&\theta p_{\theta}\xi_{1}&\theta p_{\theta}\xi_{2}&\theta p_{\theta}\xi_{3}\\
	0&-\tau(\xi_{2}q_{2}+\xi_{3}q_{3})&\tau\xi_{2}q_{1}&\tau\xi_{3}q_{1}\\
	0&\tau\xi_{1}q_{2}&-\tau(\xi_{1}q_{1}+\xi_{3}q_{3})&\tau\xi_{3}q_{2}\\
	0&\tau\xi_{1}q_{3}&\tau\xi_{2}q_{3}&-\tau(\xi_{1}q_{1}+\xi_{2}q_{2})
	\end{array}\right)
\end{equation*}
and 
\begin{equation*}
\mathcal{D}=\left(\begin{array}{cccc}
	\rho e_{\theta}v\cdot\xi-\eta\rho e_{\theta}&\xi_{1}&\xi_{2}&\xi_{3}\\
	\xi_{1}\kappa&\tau\xi\cdot v-\eta\tau&0&0\\
	\xi_{2}\kappa&0&\tau\xi\cdot v-\eta\tau&0\\
	\xi_{3}\kappa&0&0&\tau\sigma\cdot v-\eta\tau
	\end{array}\right).
\end{equation*}
Then, we have that 
\begin{equation*}
\operatorname{det}{\mathcal{A}}=\rho^{3}(\xi\cdot v-\eta)^{2}\left((\xi\cdot v-\eta)^{2}-|\xi|^{2}p_{\rho}\right)
\end{equation*}
where we set $\alpha=(\xi\cdot v-\eta)^{2}-|\xi|^{2}p_{\rho}$, so we can write 
\begin{equation*}
\operatorname{det}{\mathcal{A}}=\rho^{3}(\xi\cdot v-\eta)^{2}\alpha.
\end{equation*}
Then
\begin{equation*}
\mathcal{D}-\mathcal{C}\mathcal{A}^{-1}\mathcal{B}=\left(\begin{array}{cccc}
	\left(\rho e_{\theta}-\frac{\theta p_{\theta}^{2}}{\rho\alpha}|\xi|^{2}\right)(\xi\cdot v-\eta)&\xi_{1}&\xi_{2}&\xi_{3}\\
	\xi_{1}\kappa-\frac{p_{\theta}(v\cdot\xi-\eta)Q_{1}\tau}{\rho\alpha}&\tau(\xi\cdot v-\eta)&0&0\\
	\xi_{2}\kappa-\frac{p_{\theta}(v\cdot\xi-\eta)Q_{2}\tau}{\rho\alpha}&0&\tau(\xi\cdot v-\eta)&0\\
	\xi_{3}\kappa-\frac{p_{\theta}(v\cdot\xi-\eta)Q_{3}\tau}{\rho\alpha}&0&0&\tau(\xi\cdot v-\eta)
	\end{array}\right),
\end{equation*}
where 
	\begin{eqnarray*}
	Q_{1}&:=&\xi_{2}^{2}q_{1}+\xi_{3}^{2}q_{1}-\xi_{1}\xi_{2}q_{2}-\xi_{1}\xi_{3}q_{3},\\
	Q_{2}&:=&\xi_{1}^{2}q_{2}+\xi_{3}^{2}q_{2}-\xi_{1}\xi_{2}q_{1}-\xi_{2}\xi_{3}q_{3},\\
	Q_{3}&:=&\xi_{1}^{2}q_{3}+\xi_{2}^{2}q_{3}-\xi_{1}\xi_{3}q_{1}-\xi_{2}\xi_{3}q_{2},\\
	\end{eqnarray*}
and since
\begin{equation}
\xi_{1}Q_{1}+\xi_{2}Q_{2}+\xi_{3}Q_{3}=\mathcal{Q}(\xi;U)\xi\cdot\xi=0,
\label{eq:34}
\end{equation}
we have that
\begin{equation*}
\operatorname{det}\left(\mathcal{D}-\mathcal{C}\mathcal{A}^{-1}\mathcal{B}\right)=\left(\rho e_{\theta}-\frac{\theta p_{\theta}^{2}}{\rho\alpha}|\xi|^{2}\right)\tau^{3}(\xi\cdot v-\eta)^{4}-|\xi|^{2}\kappa\tau^{2}(\xi\cdot v-\eta)^{2}.
\end{equation*}
In consequence we have that 
\begin{equation}
\operatorname{det}\left(A(\xi;U)-\eta A^{0}(\xi;U)\right)=\rho^{3}\tau^{2}(\xi\cdot v-\eta)^{4}P_{0}(\eta),
\label{eq:35}
\end{equation}
where
\begin{equation*}
P_{0}(\eta):=\left(\alpha\rho e_{\theta}-\frac{\theta p_{\theta}^{2}}{\rho}|\xi|^{2}\right)\tau(\xi\cdot v-\eta)^{2}-\alpha|\xi|^{2}\kappa.
\end{equation*}
Since $\alpha$ is a second degree polynomial in $\eta$ then $P_{0}(\eta)$ is a fourth degree polynomial in $\eta$. This means that we will obtain four roots from the equation $P_{0}(\eta)=0$. If we introduce the variable $z:=v\cdot\xi-\eta$ we can rewrite $P_{0}(\eta)=0$ as 
\[\tau\rho e_{\theta}z^{4}-\left(\tau\rho e_{\theta}p_{\rho}+\frac{\tau\theta p_{\theta}^{2}}{\rho}+\kappa\right)z^{2}+\kappa p_{\rho}=0,\]
were we have used that $|\xi|=1$. Upon inspection of the discriminant
\begin{align*}
\Delta&=\left(\tau\rho e_{\theta}p_{\rho}+\frac{\tau\theta p_{\theta}^{2}}{\rho}+\kappa\right)^{2}-4\tau\rho e_{\theta}\kappa p_{\rho}\\
&=\left(\tau\rho e_{\theta}p_{\rho}-\kappa\right)^{2}+\frac{\tau^{2}\theta^{2}p_{\theta}^{4}}{\rho^{2}}+\frac{2\left(\tau\rho e_{\theta}p_{\rho}+\kappa\right)\tau\theta p_{\theta}^{2}}{\rho}>0,
\end{align*}
we conclude that the $z^{2}-$roots are real and positive, 
\begin{align*}
z_{+}^{2}&=\frac{1}{2}\left(p_{\rho}+\frac{\theta p_{\theta}^{2}}{\rho^{2}e_{\theta}}+\frac{\kappa}{\rho e_{\theta}\tau}\right)+\frac{1}{2}\sqrt{\left(p_{\rho}+\frac{\theta p_{\theta}^{2}}{\rho^{2}e_{\theta}}+\frac{\kappa}{\rho e_{\theta}\tau}\right)^{2}-\frac{4p_{\rho}\kappa}{\rho e_{\theta}\tau}},\\
z_{-}^{2}&=\frac{1}{2}\left(p_{\rho}+\frac{\theta p_{\theta}^{2}}{\rho^{2}e_{\theta}}+\frac{\kappa}{\rho e_{\theta}\tau}\right)-\frac{1}{2}\sqrt{\left(p_{\rho}+\frac{\theta p_{\theta}^{2}}{\rho^{2}e_{\theta}}+\frac{\kappa}{\rho e_{\theta}\tau}\right)^{2}-\frac{4p_{\rho}\kappa}{\rho e_{\theta}\tau}},
\end{align*}
with $z_{+}^{2}>z_{-}^{2}$. Hence the roots of $P_{0}(\eta)$ are
\small
\[\eta_{1}=v\cdot\xi+\sqrt{z_{+}^{2}},~
\eta_{2}=v\cdot\xi+\sqrt{z_{-}^{2}},~
\eta_{3}=v\cdot\xi-\sqrt{z_{+}^{2}}~
\mbox{and}~\eta_{4}=v\cdot\xi-\sqrt{z_{-}^{2}},\]
\normalsize
and since $\eta_{3}<\eta_{4}<\eta_{2}<\eta_{1}$, we conclude that the algebraic multiplicity of each of these roots is equal to one. They are given as
\begin{equation*}
\begin{aligned}
\eta_{1}(\xi; U)=\xi\cdot v+\frac{1}{\sqrt{2}}\sqrt{\left(p_{\rho}+\frac{\theta p_{\theta}^{2}}{\rho^{2}e_{\theta}}+\frac{\kappa}{\rho e_{\theta}\tau}\right)+\sqrt{\left(p_{\rho}+\frac{\theta p_{\theta}^{2}}{\rho^{2}e_{\theta}}+\frac{\kappa}{\rho e_{\theta}\tau}\right)^{2}-\frac{4p_{\rho}\kappa}{\rho e_{\theta}\tau}}},\\
\eta_{2}(\xi;U)=\xi\cdot v+\frac{1}{\sqrt{2}}\sqrt{\left(p_{\rho}+\frac{\theta p_{\theta}^{2}}{\rho^{2}e_{\theta}}+\frac{\kappa}{\rho e_{\theta}\tau}\right)-\sqrt{\left(p_{\rho}+\frac{\theta p_{\theta}^{2}}{\rho^{2}e_{\theta}}+\frac{\kappa}{\rho e_{\theta}\tau}\right)^{2}-\frac{4p_{\rho}\kappa}{\rho e_{\theta}\tau}}},\\
\eta_{3}(\xi;U)=\xi\cdot v-\frac{1}{\sqrt{2}}\sqrt{\left(p_{\rho}+\frac{\theta p_{\theta}^{2}}{\rho^{2}e_{\theta}}+\frac{\kappa}{\rho e_{\theta}\tau}\right)+\sqrt{\left(p_{\rho}+\frac{\theta p_{\theta}^{2}}{\rho^{2}e_{\theta}}+\frac{\kappa}{\rho e_{\theta}\tau}\right)^{2}-\frac{4p_{\rho}\kappa}{\rho e_{\theta}\tau}}},\\
\eta_{4}(\xi;U)=\xi\cdot v-\frac{1}{\sqrt{2}}\sqrt{\left(p_{\rho}+\frac{\theta p_{\theta}^{2}}{\rho^{2}e_{\theta}}+\frac{\kappa}{\rho e_{\theta}\tau}\right)-\sqrt{\left(p_{\rho}+\frac{\theta p_{\theta}^{2}}{\rho^{2}e_{\theta}}+\frac{\kappa}{\rho e_{\theta}\tau}\right)^{2}-\frac{4p_{\rho}\kappa}{\rho e_{\theta}\tau}}}.
\end{aligned}
\end{equation*}
From the formula \eqref{eq:35}, we see that the rest of the roots of the characteristic polynomial are given by the equation
\[\rho^{3}\tau^{2}(\xi\cdot v-\eta)^{4}=0,\]
implying that
\begin{equation*}
\eta_{0}(\xi;U)=v\cdot\xi,
\end{equation*}
is a root of algebraic multiplicity equal to four. Observe that $\eta_{1}(\xi;U)$, $\eta_{2}(\xi;U)$, $\eta_{3}(\xi;U)$ and $\eta_{4}(\xi;U)$ are the three dimensional analogue of the roots given in Lemma 3.2 in \cite{amp} for the one dimensional case. Now, since the algebraic multiplicity of these roots is equal to one then each one has geometric multiplicity equal to one. Meaning that, the set $\left\lbrace\eta_{i}(\xi;U)\right\rbrace_{i=1}^{4}$ provides exactly four linearly independent eigenvectors. Given that, $\eta_{0}(\xi;U)$ is different from the other roots, if the eigenvalue $\eta_{0}(\xi;U)$ had four linearly independent eigenvectors for all $U\in\mathcal{O}$ and $\xi\in\mathbb{S}^{2}$, then the inviscid Cattaneo-Christov system would be of hyperbolic nature. We will show that this is not possible for certain $U\in\mathcal{O}$ and $\xi\in\mathbb{S}^{2}$. In particular, we chose $\bar{q}$ and $\bar{\xi}$ as in Lemma 1. For such values of $\overline{U}=(\rho,v,\theta,\bar{q})\in\mathcal{O}$ and $\bar{\xi}\in\mathbb{S}^{2}$ consider the equation
\[A(\bar{\xi},\overline{U})V-\bar{\xi}\cdot v A^{0}(\overline{U})V=0,\]
where $V=(V_{1},V_{2},....,V_{8})\in\mathbb{R}^{8}$. Then, 
\begin{eqnarray}
\bar{\xi}_{1}\rho V_{2}+\bar{\xi}_{2}\rho V_{3}+\bar{\xi}_{3}\rho V_{4}&=&0,\label{eq:36}\\
\bar{\xi}_{1}p_{\rho}V_{1}+\bar{\xi}_{1}p_{\theta}V_{5}&=&0,\label{eq:37}\\
\bar{\xi}_{2}p_{\rho}V_{1}+\bar{\xi}_{2}p_{\theta}V_{5}&=&0,\label{eq:38}\\
\bar{\xi}_{3}p_{\rho}V_{1}+\bar{\xi}_{3}p_{\theta}V_{5}&=&0,\label{eq:39}\\
\theta p_{\theta}\bar{\xi}_{1}V_{2}+\theta p_{\theta}\bar{\xi}_{2}V_{3}+\theta p_{\theta}\bar{\xi}_{3}V_{4}+\bar{\xi}_{1}V_{6}+\bar{\xi}_{2}V_{7}+\bar{\xi}_{3}V_{8}&=&0,\label{eq:310}\\
-\tau(\bar{\xi}_{2}\bar{q}_{2}+\bar{\xi}_{3}\bar{q}_{3})V_{2}+\tau\bar{\xi}_{2}\bar{q}_{1}V_{3}+\tau\bar{\xi}_{3}\bar{q}_{1}V_{4}+\bar{\xi}_{1}\kappa V_{5}&=&0,\label{eq:311}\\
\tau\bar{\xi}_{1}\bar{q}_{2}V_{2}-\tau(\bar{\xi}_{1}\bar{q}_{1}+\bar{\xi}_{3}\bar{q}_{3})V_{3}+\tau\bar{\xi}_{3}\bar{q}_{2}V_{4}+\bar{\xi}_{2}\kappa V_{5}&=&0,\label{eq:312}\\
\tau\bar{\xi}_{1}\bar{q}_{3}V_{2}+\tau\bar{\xi}_{2}\bar{q}_{3}V_{3}-\tau(\bar{\xi}_{1}\bar{q}_{1}+\bar{\xi}_{2}\bar{q}_{2})V_{4}+\bar{\xi}_{3}\kappa V_{5}&=&0.\label{eq:313}
\end{eqnarray}
Multiply \eqref{eq:311} by $\bar{\xi}_{1}$, \eqref{eq:312} by $\bar{\xi}_{2}$, \eqref{eq:313} by $\bar{\xi}_{3}$ and add them up to get 
\[\kappa |\bar{\xi}|^{2}V_{5}=0,\]
implying that $V_{5}=0$. Since $p_{\rho}>0$, \eqref{eq:37} implies $V_{1}=0$. Now, notice that $V_{5}=0$ turns equations \eqref{eq:311}, \eqref{eq:312} and \eqref{eq:313} into 
\begin{eqnarray*}
-(\bar{\xi}_{2}\bar{q}_{2}+\bar{\xi}_{3}\bar{q}_{3})V_{2}+\bar{\xi}_{2}\bar{q}_{1}V_{3}+\bar{\xi}_{3}\bar{q}_{1}V_{4}&=&0,\\
\bar{\xi}_{1}\bar{q}_{2}V_{2}-(\bar{\xi}_{1}\bar{q}_{1}+\bar{\xi}_{3}\bar{q}_{3})V_{3}+\bar{\xi}_{3}\bar{q}_{2}V_{4}&=&0,\\
\bar{\xi}_{1}\bar{q}_{3}V_{2}+\bar{\xi}_{2}\bar{q}_{3}V_{3}-(\bar{\xi}_{1}\bar{q}_{1}+\bar{\xi}_{2}\bar{q}_{2})V_{4}&=&0.
\end{eqnarray*}
Observe however, that this system can be rewritten as 
\[\mathcal{Q}(\bar{\xi};\bar{q})V^{\prime}=0,\]
where $V^{\prime}=(V_{2},V_{3},V_{4})$. Hence, Lemma 1 and Remark 1 yield $V^{\prime}=\beta\bar{q}$ for some $\beta\in\mathbb{R}$. Then, according to \eqref{eq:36} and the thermodynamical assumption $\rho>0$, we get that
\[\beta\bar{\xi}\cdot\bar{q}=0.\]
But since $\bar{\xi}\cdot\bar{q}=|\bar{q}|\neq 0$, we have that $\beta=0$. Thus $V^{\prime}=0$. Then, from \eqref{eq:310} we are left with the equation
\begin{equation}
\bar{\xi}_{1}V_{6}+\bar{\xi}_{2}V_{7}+\bar{\xi}_{3}V_{8}=0.
\label{eq:314}
\end{equation}
If we set $\mathcal{M}:=\operatorname{span}{\left\lbrace(\bar{\xi_{1}},\bar{\xi_{2}},\bar{\xi_{2}})\right\rbrace}$, then $\mathbb{R}^{3}:=\mathcal{M}\oplus\mathcal{M}^{\perp}$. Hence, $\dim\mathcal{M}^{\perp}=2$. According to \eqref{eq:314}, $(V_{6},V_{7},V_{8})\in\mathcal{M}^{\perp}$. Since the rest of the components of any eigenvector $V$ of $\eta_{0}(\bar{\xi};\bar{U})$ are zero, the geometric multiplicity of $\eta_{0}(\bar{\xi};\bar{U})$ is at most two, and thus, strictly less than the algebraic multiplicity. Therefore, the three parameters left $V_{6}(\bar{\xi})$, $V_{7}(\bar{\xi})$ and $V_{8}(\bar{\xi})$, are not enough to provide four linearly independent eigenvectors associated with the eigenvalue $\eta_{0}(\bar{\xi};\bar{U})$. Thus concluding the proof.
\end{proof}
\begin{rem}
For the two dimensional case (i.e. $d=2$ and $N=6$) we can proceed in the same manner. In this case the contribution of the Lie-Oldroyd derivative yields the matrix 
\begin{equation*}
\mathcal{Q}_{2}(\xi;q)=\left(\begin{array}{cc}
	-\xi_{2}q_{2}&\xi_{2}q_{1}\\
	\xi_{1}q_{2}&-\xi_{1}q_{1}
\end{array}\right),
\end{equation*}
which again, by taking $\bar{q}_{1}\neq \bar{q}_{2}\neq 0$ and $\bar{\xi}=\frac{\bar{q}}{|\bar{q}|}$ we get that 
\[\dim\ker \mathcal{Q}_{2}(\bar{\xi};\bar{q})=1.\]
The characteristic polynomial will be given as 
\begin{equation*}
\operatorname{det}\left(A(\xi;U)-\eta A^{0}(\xi;U)\right)=\rho^{2}\tau(\xi\cdot v-\eta)^{2}P_{0}(\eta),
\end{equation*}
where $P_{0}(\eta)$ is the same as in the three dimensional case. Hence $\eta_{1}(\xi;U)$, $\eta_{2}(\xi;U)$, $\eta_{3}(\xi;U)$ and $\eta_{4}(\xi;U)$ will appear one more time as eigenvalues of algebraic multiplicity equal to one. The only difference is that now $\eta_{0}(\xi;U)=v\cdot\xi$ will be an eigenvalue of algebraic multiplicity equal to two. Still, by the same procedure described in theorem $(3.1)$, we can show that the geometric multiplicity of this eigenvalue is equal to one. Thus, we conclude that the two dimensional version of the Cattaneo-Christov system is not hyperbolic.
\end{rem}
\section{Non-symmetrizability}

The hyperbolicity of a quasilinear system of partial differential equations is related with other relevant mathematical properties for system \eqref{eq:21}. One of them is its \emph{symmetrizability}. We remind the reader that a system of the form \eqref{eq:21} is said to be symmetrizable if there exists, $S(U)$, a matrix value function of $U\in\mathcal{O}$, of order $N\times N$, that is smooth and positive definite, such that $S(U)A^{0}(U)$, $S(U)A^{i}(U)$ and $S(U)D(U)$ are symmetric matrices for all $i,j=1,2,3$ and $S(U)A^{0}(U)$ is positive definite (see \cite{benzoni}).\\
For the system \eqref{eq:21}, there is no standard way to compute a symmetrizer without assuming that the system is derived from a set of viscous balance laws. This is not the case for the inviscid Cattaneo-Christov system given that Christov's evolution equation, \eqref{eq:5}, is not given in conservative form. In particular, we cannot rely on the existence of a convex entropy function to assure the existence of the symmetrizer \cite{ka}. However, it is well-known that symmmetrizable systems are hyperbolic (\cite{benzoni}, \cite{serre}), thus, from theorem \ref{nonhyp} it follows that
\begin{coro}
Under the thermodynamical assumptions \textbf{C2}-\textbf{C4}, there is no symmetrizer for the inviscid Cattaneo-Christov system in three and two space dimensions.\qed
\end{coro}
There are three main reasons why it is convenient to find a symmetrizer for systems of the form \eqref{eq:21}:
\begin{itemize}
	\item [(a)] It implies the hyperbolicity.
	\item [(b)] Is essential when looking for energy estimates for the Cauchy problem of  \eqref{eq:21} (see \cite{benzoni}, \cite{otto} and \cite{serre}). In order to establish the local well-possednes of the Cauchy problem, the symmetrizer is needed to construct the \emph{energy} of the equation ( see \cite{katosym}, \cite{kawa}, \cite{lax}, \cite{maj}, \cite{matsu}, \cite{novo}, \cite{rack}, amongst others).
	\item [(c)] Establishing the strict dissipativity of the linearized system (Kawashima's genuinely coupling condition, see \cite{ka1} and \cite{hump}). In turn, under the assumption of symmetrizability, the strict dissipativity is equivalent to the existence of global decay rates, according to  Kawashima's equivalence theorem \cite{ka1}.
\end{itemize}
For linear equations with constant coefficients, it is well-known that the hyperbolicity is a necessary and sufficient condition to assure the well-posedness of the Cauchy problem (\cite{otto}, \cite{serre}). In general this means that the linearization of a system of the form \eqref{eq:21} around constant states $U_{c}\in\mathcal{O}$ requires the hyperbolicity of the quasilinear system in order to be well-posed. However, for the inviscid Cattaneo-Christov system, if the constant state $U_{c}$ is an \emph{equilibrium state} the linearization around it will be well-posed. We say that $U_{c}=(\rho_{c},v_{c},\theta_{c},q_{c})^{\top}\in\mathcal{O}$ is a constant equilibrium state if $U_{c}$ is constant and $q_{c}=0$. Let us consider a constant state $U_{c}\in\mathcal{O}$ and assume that $U_{c}+W$ is a solution of \eqref{eq:21}. Then we can recast \eqref{eq:21} as 
\begin{equation}
A^{0}(U_{c})W_{t}+A^{i}(U_{c})\partial_{i}W+Q(U_{c})+DQ(U_{c})W+\mathcal{N}(W,D_{x}W,D_{x}^{2}W,W_{t})=0,
\label{eq:41}
\end{equation} 
where $\mathcal{N}(W,D_{x}W,D_{x}^{2}W,W_{t})$ comprises all the nonlinear terms and $DQ(U)$ denotes the Jacobian matrix of $Q=Q(U)$, which is a diagonal matrix given as 
\[DQ(U)=\left(\begin{array}{cc}
	\mathbb{O}_{5}&\\
	&\mathbb{I}_{3}
\end{array}\right),\]
for all $U\in\mathcal{O}$. Here $\mathbb{O}_{5}$ denotes the zero matrix of order $5\times 5$. If we only consider the linear part of $\eqref{eq:41}$ and $U_{c}$ is taken as constant equilibrium state, then $Q(U_{c})=0$, and so we get the linear constant coefficient system
\begin{equation}
A^{0}(U_{c})W_{t}+A^{i}(U_{c})\partial_{i}W+DQ(U_{c})W=0.
\label{eq:42}
\end{equation}
Notice that for this system $\mathcal{Q}(\xi;q_{c})=0$, that is, for a constant equilibrium state, the symbol $A(\xi;U_{c})$ has no terms involving $q_{c}$. Thus, for this case, Lemma 1 is unsatisfied. In fact, system \eqref{eq:42} is symmetrizable and thus hyperbolic. In this case, a constant symmetrizer can be obtained by making $U=U_{c}$ in the matrix function
\begin{equation}
S(U)=\left(\begin{array}{cccc}
\frac{p_{\rho}}{\rho}&&&\\
&\mathbb{I}_{3}&&\\
&&\frac{1}{\theta}&\\
&&&\frac{1}{\kappa\theta}\mathbb{I}_{3}\\
\end{array}\right),
\label{eq:43}
\end{equation}
defined for any $U\in\mathcal{O}$. Hence, the Cauchy problem for \eqref{eq:42} is well-posed (see \cite{otto}). On the other hand, by taking constant states with the property $q_{1c}\neq q_{2c}\neq q_{3c}\neq 0$ (as in Lemma 1) the linear part of \eqref{eq:41} will not be hyperbolic and its Cauchy problem will be ill-possed.
\section{The Cattaneo-Christov-Jordan system}
Consider the coupling between equations \eqref{eq:1}-\eqref{eq:3} and the evolution equation \eqref{eq:4} in three space dimensions ($d=3$ and $N=8$). This model was first proposed by Christov and Jordan (see, \cite{chjo}). The equations \eqref{eq:1}-\eqref{eq:4} can be written in the quasilinear form \eqref{eq:21} with the same $A^{0}(U)$ and $Q(U)$ for $U\in\mathcal{O}$ as in the inviscid Cattaneo-Christov system. In fact, the only difference will be that the sub-matrix $\mathcal{Q}(\xi;U)$ for this new system will be zero, $\mathcal{Q}(\xi;U)=0$. That is, the new symbol $A(\xi;U)$ will not have terms involving $q$. We will refer to this system as the inviscid Cattaneo-Christov-Jordan system for compressible fluid flow.\\
For the  inviscid Cattaneo-Christov-Jordan system, it can easily be verified that $S=S(U)$ is a symmetrizer. Thus, this system is hyperbolic. Moreover, its characteristic polynomial, will be exactly the same as the characteristic polynomial for the inviscid Cattaneo-Christov system. To verify this fact we can proceed exactly as in the computations made during the proof of \ref{nonhyp}. The slight difference being that now the terms $Q_{1}$, $Q_{2}$ and $Q_{3}$ will not appear in the matrix $\mathcal{D}-\mathcal{C}\mathcal{A}^{-1}\mathcal{B}$, that is, for practical purposes we can consider them to be zero for the Cattaneo-Christov-Jordan system, and so, relation \eqref{eq:34} will be trivially satisfied. Thus, the term
\[\xi_{1}Q_{1}+\xi_{2}Q_{2}+\xi_{3}Q_{3}\]
will have no effect on the computation of the determinant of $\mathcal{D}-\mathcal{C}\mathcal{A}^{-1}\mathcal{B}$ , just as in the inviscid Cattaneo-Christov system. Hence, the characteristic polynomial of the Cattaneo-Christov-Jordan system will be given by \eqref{eq:35}. In particular, this means that the the Cattaneo-Christov-Jordan system and the Cattaneo-Christov system share the same eigenvalues $\eta_{0}(\xi;U)$, $\eta_{1}(\xi;U)$, $\eta_{2}(\xi;U)$, $\eta_{3}(\xi;U)$ and $\eta_{4}(\xi;U)$. But in the case of the Cattaneo-Christov-Jordan system, these eigenvalues represent the characteristic speeds of eight linearly independent wave fronts travelling in the direction $\xi\in\mathbb{S}^{2}$. Also, notice that the linearization of the Cattaneo-Christov-Jordan system around a constant equilibrium state, $U_{c}\in\mathcal{O}$, coincides with the linearization of the inviscid Cattano-Christov system given in \eqref{eq:42}.\\
Finally, since the Cattaneo-Christov-Jordan system is symmetrizable, its Cauchy problem is locally well-posed (see, \cite{katosym}, \cite{otto} and \cite{lax}).

\section{Discussion and conclusions}
The aim of the coupling between the Cattaneo-Christov evolution equation \eqref{eq:5} and the equations \eqref{eq:1}-\eqref{eq:3} was to obtain a hyperbolic heat conducting theory describing the dynamics of fluid flow that is compatible with the frame indifference principle in continuum dynamics. On one hand, it is true that \eqref{eq:5} is a frame invariant formulation of the Maxwell-Cattaneo law (\cite{christov}) that can be used together with the equation
\[\rho c_{p}\left(\theta_{t}+v\cdot\nabla\theta\right)=-\nabla\cdot q\]
(a simplified version of the equation for the energy \eqref{eq:3} where $c_{p}$ is the specific heat), to derive a single hyperbolic equation for the temperature field. But, as we have shown in this work, the inviscid Cattaneo-Christov system is not hyperbolic. On the other hand, Euler's equation for a compressible fluid are of hyperbolic nature, once assumed \textbf{C2}, \textbf{C4} and \textbf{C5}, but once we use Fourier's heat flux law to derive a single equation for the temperature we obtain a parabolic equation. So, it seems that, it is the given constitutive law for the heat flux that ultimately decides the kind of heat equation that will remain. For this reason, there are two different concepts at play here, the hyperbolicity (or lack there of) of the coupled system, and the hyperbolicity (or lack there of) of the second order in time equation for the temperature. As we showed, 
the introduction of the Cattaneo-Christov model doesn't focuses in the hyperbolicity of the system \eqref{eq:21} but rather, as Christov showed, focuses only in the hyperbolicity of the second order equation for the temperature.\\
Now, given that the hyperbolicity is a necessary condition for the well-posedness of the noncharacterisitc Cauchy problem for nonlinear partial differential equations (\cite{guy}), we cannot assure that the inviscid Cattaneo-Christov system properly describes the dynamics of heat conducting compressible fluid flow. The lack of a symmmetrizer for this system prevents us to apply the standard well-posedness $L^{2}$-theory for both linearized and quasilinear cases. Moreover, one can verify Kawashima's genuinely coupling condition for the linear part of \eqref{eq:41}, but it has no mathematical relevance since Kawashima's equivalence theorem (\cite{ka1}) cannot be verified.\\
Another drawback of the non-existence of the symmetrizer for the inviscid Cattaneo-Christov system is not being able to use the standard theory of conservation laws to study it. Even if there is a diffeomorphism $U\mapsto V$, $V=V(U)\in\mathbb{R}^{N}$, from $\mathcal{O}$ into an open set $\mathcal{O}_{V}$, $N$-dimensional vector functions $F^{j}$, $j=0,1,..,d$, $N\times N$ matrix functions $G^{ij}$ $i,j=1,..,d$ depending smoothly on $V$, such that \eqref{eq:21} becomes 
\[V_{t}+\sum_{j=1}^{d}\partial_{j}F^{j}(V)+F^{0}(V)=\sum_{i,j=1}^{d}\partial_{i}\left(G^{ij}(V)\partial_{j}V\right),\]
there will not exist a convex entropy function (see, \cite{daf}, \cite{kurt} and \cite{ka}). The existence of a convex entropy function plays a fundamental role in establishing admissibility conditions for selecting, from all weak solutions of a conservation law, that which is stable from the physical or mathematical point of view (see \cite{serre}). For this reason, even in the case in which the Cattaneo-Christov system can be rewritten as a conservation law, the non-existence of a convex entropy function is a major disadvantage for the theory of weak solutions for this system of equations. Is important to mention that the author of this work has not being able to find the change of variables $V=V(U)$, not even in one space dimension.\\
We showed that for constant equilibrium states the linear part of \eqref{eq:41} is symmetrizable and thus hyperbolic, we recognized that this linearization coincides with the linearization of a different quasilinear system of equations, the so called inviscid Cattaneo-Christov-Jordan system. This system is symmetrizable, hence hyperbolic and, in consequence, its Cauchy problem is well-posed. Furthermore, notice that in one space dimension, the inviscid Cattaneo-Christov system and the inviscid Cattaneo-Christov-Jordan system are the same. For this reason, in several space dimensions, the inviscid Cattaneo-Christov-Jordan system is the exact analogue of the one dimensional inviscid Cattaneo-Christov system. One outstanding fact is that both systems share the same eigenvalues as consequence of sharing the same characteristic polynomial. Now, the characteristic polynomial \eqref{eq:35} does not depends on $q$, this is a striking fact, given that the symbol $A(\xi;U)$ of the Cattaneo-Christov system has a block matrix of order $3\times 3$ filled with the components of $q$, i.e. $\tau\mathcal{Q}(\xi;U)$. Since this block matrix was induced by the Lie-Oldroyd derivative in Christov's evolution equation \eqref{eq:5}, then this frame invariant formulation has no effects in the characteristic speeds of heat propagation in a compressible inviscid fluid flow. Thus supporting the case for the Cattaneo-Christov-Jordan system.\\
This feature of Christov's equation is a consequence of relation \eqref{eq:34}. Thus, if any type of objective derivative is used to derive a Cattaneo type model for the heat flux and then is to be coupled with equations \eqref{eq:1}-\eqref{eq:3}, as long as \eqref{eq:34} is satisfied, the characteristic speeds of this new system will be the same as the the ones of the Cattaneo-Christov-Jordan system.\\
We conclude that, compared to the standard material derivative, the introduction of the Lie-Oldroyd material derivative to obtain a frame invariant formulation of the Cattaneo law for the heat flux seems to serve the only computational purpose of solving for a single hyperbolic differential equation for the temperature. Even if this feature is considered an advantage, it brings a major mathematical and physical drawback, the non-hyperbolicity of the inviscid Cattaneo-Christov system.\\
Under the light of these observations, we conclude that our result seems to jeopardize the applicability of the inviscid Cattaneo-Christov system to properly describe the motion of compressible, thermally relaxed fluids in several space dimensions.
\section*{Acnowledgments}
Thanks to professor Ram\'on G. Plaza for his comments and guidance. This work was supported by CONACyT (M\'exico), through a scholarship for doctoral studies, grant no. 465484 and by DGAPA-UNAM program PAPIIT, grant IN100318.

\bibliographystyle{plain} 
\bibliography{feVSkawabiblio}

\end{document}